\documentclass[10pt, conference]{IEEEtran}
\usepackage{times}
\usepackage{cite}
\usepackage{xcolor}
\usepackage[cmex10]{amsmath}
\usepackage{mathtools}
\usepackage{amssymb}
\usepackage{epsfig,verbatim}
\usepackage{algorithm}
\usepackage{algpseudocode}

\providecommand{\Ac}{{\cal A}}
\providecommand{\Cc}{{\cal C}}
\providecommand{\Dc}{{\cal D}}

\providecommand{\Ic}{{\cal I}}
\providecommand{\Lc}{{\cal L}}

\newtheorem{MyTheorem}{Theorem}

\begin{document}
\title{Capacity-Achieving Rate-Compatible Polar Codes for General Channels}

\author{\IEEEauthorblockN{Marco~Mondelli}
\IEEEauthorblockA{EPFL, Switzerland\\
marco.mondelli@epfl.ch}
\and
\IEEEauthorblockN{S.~Hamed~Hassani}
\IEEEauthorblockA{ETH Z\"{u}rich, Switzerland\\
hamed@inf.ethz.ch}
\and
\IEEEauthorblockN{Ivana Mari\'c, Dennis Hui}
\IEEEauthorblockA{Ericsson Research, Santa Clara, USA\\
\{ivana.maric, dennis.hui\}@ericsson.com}
\and
\IEEEauthorblockN{Song-Nam Hong}
\IEEEauthorblockA{Ajou University, Korea\\
snhong@ajou.ac.kr}
}

%\author{Marco~Mondelli, S.~Hamed~Hassani, Ivana Mari\'c, Dennis Hui, Song-Nam Hong
%\thanks{M.~Mondelli is with the School of Computer and Communication Sciences, EPFL, Switzerland (email: marco.mondelli@epfl.ch)}
%\thanks{S.~H.~Hassani is with the Computer Science Department, ETH Z\"{u}rich, Switzerland (email: hamed@inf.ethz.ch)}
%\thanks{I.~Maric and D. Hui are with Ericsson Research, San Jose, CA, USA (emails: \{ivana.maric, dennis.hui\}@ericsson.com)}
%\thanks{S.-N. Hong is with Ajou University, Korea (email: )}
%}

\maketitle

%%%%%%%%%%%%%%%%%%
\begin{abstract}
We present a rate-compatible polar coding scheme that achieves the capacity of any family of channels. Our solution generalizes the previous results \cite{HongHuiMaric2015, LiTse2015} that provide capacity-achieving rate-compatible polar codes for a \emph{degraded} family of channels. The motivation for our extension comes from the fact that in many practical scenarios, e.g., MIMO systems and non-Gaussian interference, the channels cannot be ordered by degradation. The main technical contribution of this paper consists in removing the degradation condition. To do so, we exploit the ideas coming from the construction of universal polar codes. 

Our scheme possesses the usual attractive features of polar codes: low complexity code construction, encoding, and decoding; super-polynomial scaling of the error probability with the block length; and absence of error floors. On the negative side, the scaling of the gap to capacity with the block length is slower than in standard polar codes, and we prove an upper bound on the scaling exponent. 
\end{abstract}

%%%%%%%%%%%%%%%%%%%%%%%%
\section{Introduction}
%%%%%%%%%%%%%%%%%%%%%%%%
% 1. Introduction
%%%%%%%%%%%%%%%%%%%%%%%
Polar codes, introduced by Ar{\i}kan \cite{Ari09}, achieve the capacity of any binary memoryless symmetric (BMS) channel with encoding and decoding complexity $\Theta(n \log_2 n)$, where $n$ is the block length of the code. The code construction can be performed with complexity $\Theta(n)$ \cite{TV13con, RHTT}. Furthermore, a unified characterization of the performance of polar codes in several regimes is provided in \cite{MHU15unif-ieeeit}. Here, let us just recall the following basic facts: the error probability scales with the block length roughly as $2^{-\sqrt{n}}$ \cite{ArT09}; the gap to capacity scales with the block length as $n^{-1/\mu}$, and bounds on the scaling exponent $\mu$ are provided in \cite{HAU14, GB14, MHU15unif-ieeeit}; and polar codes are not affected by error floors \cite{MHU15unif-ieeeit}. A successive cancellation list (SCL) decoder with space complexity $O(L n)$ and time complexity $O(L n \log_2 n)$ is proposed in \cite{TVa15}, where $L$ is the size of the list. Empirically, the use of several concurrent decoding paths yields an error probability comparable to that under optimal MAP decoding with practical values of the list size. In addition, by adding only a few extra bits of cyclic redundancy check (CRC) precoding, the resulting performance is comparable with state-of-the-art LDPC codes. Because of their attractive features, polar codes are being considered for use in future wireless communication systems (e.g., 5G).
% (e.g. 5G cellular systems).

%%%%%%%%%%%%%%%%%
% 2. HARQ-IR
%%%%%%%%%%%%%%%%%
Wireless communication systems require adaptive transmission techniques due to the time-varying nature of the channel. A popular solution is based on hybrid automatic repeat request with incremental redundancy (HARQ-IR). The idea is to send parity bits in an incremental fashion depending on the quality of the time-varying channel. Such systems require the use of \emph{rate-compatible} codes typically obtained by puncturing. In a rate-compatible family of codes the set of parity bits of a code with higher rate is a subset of the set of parity bits of a code with lower rate. In this way, if decoding is not successful at a particular rate, then the receiver can request only the additional parity bits from the transmitter, instead of the full set of parity bits of the code with lower rate.

% The construction of rate-compatible turbo codes and LDPC codes has been the subject of extensive research (see \cite{Rowitch2000, Tsung-Yi2015} and references therein). 
%%%%%%%%%%%%%%%%%
% 3. Motivation
%%%%%%%%%%%%%%%%%

%Although polar codes can achieve the capacity of any symmetric binary-input channel, designing good rate-compatible constructions is more challenging. Puncturing of polar codes yields a rate loss, which means that the resulting scheme is not capacity-achieving.

%As in \cite{Eslami2011, Shin2013, WangLiu2014, Miloslavskaya2015} an information set is optimized according to a puncturing pattern, these methods cannot be used to design a family of rate-compatible punctured codes as required for HARQ-IR, where the same information set (generally optimized for the mother code) should be used for all punctured codes in the family. In \cite{Huawei2014}, a heuristic search algorithm was presented to design a good puncturing pattern for a fixed information set.

% However, finding an optimal rate-compatible puncturing pattern with low complexity is still an open problem. \textcolor{red}{Here, you refer to the fact that none of these methods is capacity-achieving?}

%%%%%%%%%%%%%%%%%%%%%%%%%%%%
% 3. Our approach
%%%%%%%%%%%%%%%%%%%%%%%%%%%% 
The study of puncturing patterns for polar codes is considered in \cite{Eslami2011, Shin2013, Huawei2014, WangLiu2014}, and an efficient algorithm for optimizing jointly the puncturing pattern and the set of information bits of the code is proposed in \cite{Miloslavskaya2015}. When the channels are degraded, the authors of \cite{HongHuiMaric2015, LiTse2015} propose rate-compatible polar-like codes that are provably capacity-achieving. These solutions take advantage of the nested property of polar codes for degraded channels \cite[Lemma 4.7]{Kor09thesis}. Hence, they can be used only when the family of channels over which the transmission takes place is ordered by degradation, otherwise the nested property does not hold. Let us mention two notable practical scenarios in which the degradation condition is not satisfied: (i) When the transmitter and/or receiver have multiple antennas (MIMO systems), the family of transmission channels is not degraded. (ii) When there is non-Gaussian interference at the receiver (e.g., another user is transmitting a modulated signal), the family of channels is again not degraded. 
In such communications scenarios, the schemes proposed in \cite{HongHuiMaric2015, LiTse2015} do not achieve capacity.

In this paper, we present rate-compatible and provably capacity-achieving polar codes \emph{for any family of BMS channels}. In order to remove the degradation assumption, we take advantage of the ideas coming from the construction of universal polar codes \cite{HRunipol, SaWa16}. These techniques have proven useful to align the polarized indices for the wiretap channel in \cite{SaV13}, for the broadcast channel in \cite{MHSU14broad-ieeeit} and, later on, for many other non-standard communication scenarios. We also show that our solution retains the usual attractive properties of polar codes: low complexity for code construction, encoding, and decoding; error probability that scales with the block length roughly as $2^{-\sqrt{n}}$; and absence of error floor. However, the scaling of the gap to capacity with the block length is slower than that of standard polar codes, and we provide an upper bound on the scaling exponent of rate-compatible polar codes. 

The remainder of this paper is organized as follows. In Section \ref{sec:prel}, we present some basic facts about rate-compatible codes and about polar codes. Furthermore, we review the scheme of \cite{HongHuiMaric2015} to construct capacity-achieving rate-compatible polar codes for a degraded family of channels. In Section \ref{sec:mainres}, we describe the proposed rate-compatible family of polar codes and  analyze its properties. In Section \ref{sec:concl}, we provide  concluding remarks.

% Due to the space limitation, the proofs are omitted. Detailed proofs for all theorems can be found in \cite{HongHuiMaric2015}.

%%%%%%%%%%%%%%%%%%%%%%%%%%%%
\section{Preliminaries} \label{sec:prel}
%%%%%%%%%%%%%%%%%%%%%%%%%%%%

%%%%%%%%%%%%%%%%%%%%%%%%%%%%
\subsection{Rate-Compatible Codes} \label{subsec:ratecomp}
%%%%%%%%%%%%%%%%%%%%%%%%%%%%

To simplify notation, let $[K] \triangleq \{1,2,\ldots,K\}$ for any positive integer $K$ and, for $i<j$, let $a^{i:j}$ be a shortcut for the vector $(a^{(i)}, \ldots, a^{(j)})$. Given a fixed number of information bits $k$, consider a family of codes $\{\Cc_{1}, \Cc_{2}, \ldots, \Cc_{K}\}$ with block lengths $\bar{n}_1 < \bar{n}_2 < \ldots <\bar{n}_K$ and rates $R_1 > R_2 > \ldots > R_K$ such that $R_i = k/\bar{n}_i$ for $i\in [K]$. We say that $\{\Cc_{1}, \Cc_{2}, \ldots, \Cc_{K}\}$ is {\it rate-compatible} if the codewords of the code $\Cc_{i}$ can be obtained by removing $\bar{n}_j-\bar{n}_i$ bits from the codewords of the code $\Cc_{j}$ for any $j>i$. 

Let us now explain how to use a family of rate-compatible codes $\{\Cc_{1}, \Cc_{2}, \ldots, \Cc_{K}\}$ for HARQ-IR. Given the information vector $u^{1:k}\in\{0,1\}^k$, we first transmit a codeword $x_1^{1:\bar{n}_1}$ of $\Cc_{1}$. At the receiver, if the decoding is successful, the procedure ends; otherwise, an error message (NACK) is sent back to the transmitter. As the family of codes is rate-compatible, the codeword $x_2^{1:\bar{n}_2}$ of $\Cc_{2}$ can be obtained by adding $\bar{n}_2-\bar{n}_1$ extra bits to $x_1^{1:\bar{n}_1}$. As $x_1^{1:\bar{n}_1}$ has been already sent, in order to transmit $x_2^{1:\bar{n}_2}$, we need to send only these $\bar{n}_2-\bar{n}_1$ extra bits over the channel. Again, at the receiver, if the decoding is successful, the procedure ends; otherwise, an error message (NACK) is sent back to the transmitter. In general, a codeword $x_{i+1}^{1:\bar{n}_{i+1}}$ of $\Cc_{i+1}$ can be obtained by adding $\bar{n}_{i+1}-\bar{n}_i$ bits to a codeword $x_i^{1:\bar{n}_i}$ of $\Cc_{i}$. Hence, if $x_i^{1:\bar{n}_i}$ is not recovered correctly, in order to transmit $x_{i+1}^{1:\bar{n}_{i+1}}$, we need to send only these $\bar{n}_{i+1}-\bar{n}_i$ extra bits over the channel.

Set $\bar{n}_0 = 0$ and $n_i = \bar{n}_i-\bar{n}_{i-1}$, for $i\in [K]$. Then, we have that $\bar{n}_i = \sum_{j=1}^i n_j$, and we will refer to $\{n_1, n_2, \ldots, n_K\}$ as the set of incremental block lengths of the family of codes $\{\Cc_{1}, \Cc_{2}, \ldots, \Cc_{K}\}$. Note that the condition on the rates can be rewritten as 
\begin{equation}\label{eq:conditionrate}
R_i = \frac{k}{\sum_{j=1}^i n_j}, \qquad \forall \hspace{0.2em}i \in [K].
\end{equation}
Let $W_1, W_2, \ldots, W_K$ denote a family of $K$ BMS channels with capacities $I(W_1) > I(W_2) > \ldots > I(W_K)$. We say that a sequence of rate-compatible families of codes $\{\Cc_{1, m}, \Cc_{2, m}, \ldots, \Cc_{K, m}\}_{m\in \mathbb{N}}$, designed for a monotonically increasing sequence of information sizes $\{k_m\}_{m\in \mathbb{N}}$, achieves the capacity of $W_1, W_2, \ldots, W_K$ if, for any $i\in [K]$, the sequence of codes $\{\Cc_{i, m}\}_{m\in \mathbb{N}}$ achieves the capacity of $W_i$. In other words, we require that, for any $i\in [K]$, the sequence of codes  $\{\Cc_{i, m}\}_{m\in \mathbb{N}}$ has block lengths $\bar{n}_{i, m}\to \infty$, rates $R_{i, m}$ converging to a value arbitrarily close to $I(W_i)$, and vanishing error probability when the transmission takes place over $W_i$. To avoid cluttering, for the rest of this paper we will drop the index $m$ when considering a sequence of rate-compatible families of codes.

%%%%%%%%%%%%%%%%%%%%%%%%%%%%
\subsection{Polar Codes for Transmission over a BMS Channel} \label{Polarcodes}
%%%%%%%%%%%%%%%%%%%%%%%%%%%%

Consider transmission over a BMS channel $W$ with capacity $I(W)$, and let $X$ and $Y$ denote the input and the output of the channel, respectively. In the following, we briefly revise how to transmit over $W$ with a rate arbitrarily close to $I(W)$ by using polar codes.
%First, we give a high-level description of the scheme. Then, we provide a performance analysis for the error exponent, scaling exponent, and error floor regimes. 

\noindent {\bf Design of the Scheme.} Let $n\in \mathbb N$ be a power of $2$ and consider the $n\times n$ matrix $G_n$ defined as
\begin{equation}\label{eq:defFtimesn}
G_n = B_n F^{\otimes \log_2 n}, \qquad \qquad F =  \biggl[ \begin{array}{cccc}
1 & & & 0 \\
1 & & & 1 \end{array} \biggr], 
\end{equation}
where $F^{\otimes \log_2 n}$ denotes the $\log_2 n$-th Kronecker power of $F$, and $B_n$ is the permutation matrix that acts as a bit-reversal operator (see Section VII-B of \cite{Ari09} for further details).

Let $X^{1:n}$ be a vector with $n$ i.i.d. uniformly random components and define $U^{1:n} = X^{1:n} G_n$. Consider the set 
\begin{equation} \label{eq:lhsets}
\begin{split}
\Lc_{X \mid Y} &= \{i \in [n] \colon Z(U^{(i)} \mid  U^{1:i-1}, Y^{1:n}) \le \delta_n\},
\end{split}
\end{equation}
where $Z(\cdot | \cdot)$ denotes the Bhattacharyya parameter that is defined as follows. Given $(T, V)\sim p_{T, V}$, with $T$ binary and $V$ taking values in
an arbitrary discrete alphabet ${\mathcal V}$, we set
\begin{equation} \label{eq:BattaTV}
Z(T\mid V) = 2 \sum_{v\in{\mathcal V}}
{\mathbb P}_V(v)\sqrt{{\mathbb P}_{T\mid V}(0\mid v){\mathbb P}_{T\mid V}(1\mid v)}.
\end{equation}
The Bhattacharyya parameter $Z(T\mid V)$ is close to $0$ (or $1$) if and only if the conditional entropy $H(T\mid V)$ is 
close to $0$ (or $1$). Consequently, if $Z(T\mid V)$ is close to $0$, then $T$ is approximately a deterministic function of $V$ and, 
if $Z(T\mid V)$ is close to $1$, then $T$ is approximately uniformly
distributed and independent of $V$. Hence, for $i\in \Lc_{X \mid Y}$, the bit $U^{(i)}$ is approximately a deterministic function of the previous bits $U^{1:i-1}$ and the channel output $Y^{1:n}$. This means that $U^{1:n}$ can be decoded in a successive fashion, given the channel output $Y^{1:n}$ and the values $\{U^{(i)}\}_{i\in \Lc_{X \mid Y}^{\rm c}}$, where $\Lc_{X \mid Y}^{\rm c}$ denotes the complement of $\Lc_{X \mid Y}$. Furthermore,
\begin{equation} \label{eq:cardHL}
\begin{split}
\lim_{n\to \infty}\frac{1}{n} \, | \Lc_{X \mid Y}|  &= I(W).\\
\end{split}
\end{equation}

In order to construct a polar code for the channel $W$, we proceed as follows. The information bits are placed in the positions of $U^{1:n}$ indexed by $\Lc_{X \mid Y}$, as these positions will be decodable in a successive fashion given the output. The remaining positions of $U^{1:n}$ are frozen and their values are shared between the encoder and the decoder. Any choice of the frozen bits is as good as any other (see Section VI-B of \cite{Ari09}). Hence, for the sake of simplicity, we can simply set these bits to $0$. Finally, we transmit over the channel the vector $X^{1:n} = U^{1:n} G_n^{(-1)} = U^{1:n} G_n$, where the last equality follows from the fact that $G_n = G_n^{(-1)}$. 

%\noindent {\bf Encoding.} We place the information into the positions indexed by $\Lc_{X \mid Y}$, hence let $\{u^{(i)}\}_{i \in \Lc_{X \mid Y}}$ denote the information bits to be transmitted. The remaining positions are filled with all $0$s. Note that, by \eqref{eq:cardHL}, the requirement on the transmission rate is met. As $G_n = G_n^{-1}$, the vector $x^{1:n} = u^{1:n} G_n$ is transmitted over the channel. As discussed in Section VII of \cite{Ari09}, by exploiting the particular structure of the matrix $G_n$, it is possible to perform this matrix multiplication with complexity $\Theta(n \log_2 n)$.

%\noindent {\bf Decoding.} The decoder receives $y^{1:n}$ and computes the estimate $\hat{u}^{1:n}$ of $u^{1:n}$ according to the rule
%\begin{equation}\label{eq:decrulecc}
%\hat{u}^{(i)} \hspace*{-0.3em}=\hspace*{-0.3em} \left\{ \begin{array}{ll}
%\hspace*{-0.3em}\displaystyle\arg\max_{u\in \{0, 1\}} \hspace*{-0.3em}{\mathbb P}_{U^{(i)}
%| U^{1:i-1}\hspace*{-0.1em}, Y^{1:n}}(u | u^{1:i-1}\hspace*{-0.4em}, y^{1:n}), & 
%\mbox{if } i \hspace*{-0.1em}\in\hspace*{-0.1em} \Lc_{X \mid Y},\\
%u^{(i)}, & \mbox{ otherwise,} \\
% \end{array}\right.
%\end{equation}
%where the probabilities ${\mathbb P}_{U^{(i)} \mid U^{1:i-1}, Y^{1:n}}(u \mid u^{1:i-1}, y^{1:n})$ can be computed recursively with complexity $\Theta(n \log_2 n)$ (see Section VIII of \cite{Ari09}). This is a successive cancellation decoder, as the estimates are produced one by one with a single pass on the data (as opposed to iterative decoding).

\noindent {\bf Complexity and Performance Analysis.} Let us first discuss the \emph{code construction}. The problem consists in finding the set $\Lc_{X \mid Y}$, which is equivalent to computing $Z(U^{(i)} \mid  U^{1:i-1}, Y^{1:n})$ for $i\in [n]$. The approximate computation of these Bhattacharyya parameters (hence, the code construction) can be performed in $\Theta(n)$ by using the techniques described in~\cite{TV13con, RHTT}.

Let us now consider the \emph{encoding and decoding complexity}. As discussed in Section VII of \cite{Ari09}, by exploiting the particular structure of the matrix $G_n$, it is possible to perform the matrix multiplication $U^{1:n} G_n$ with complexity $\Theta(n \log_2 n)$. Similarly, the decoding complexity is $\Theta(n \log_2 n)$ (see Section VIII of \cite{Ari09} for further details).

Finally, let us discuss the \emph{error performance}. The block error probability $P_{\rm B}$ can be upper bounded by the sum of the Bhattacharyya parameters of the channels that are not frozen (see Proposition 2 of \cite{Ari09}). In formulae, 
\begin{equation}\label{eq:ubPe}
P_{\rm B} \le \sum_{i\in {\Lc_{X \mid Y}}} Z(U^{(i)} \mid U^{1:i-1}, Y^{1:n})\le n \delta_n.
\end{equation}
In Ar{\i}kan's original paper~\cite{Ari09}, $\delta_n$ is upper bounded by $n^{-5/4}$, hence $P_{\rm B}$ is $O(n^{-1/4})$. This bound is refined in \cite{ArT09}, where it is shown that $P_{\rm B}$ is $O(2^{-n^{\beta}})$ for any $\beta \in (0, 1/2)$. This means that the error probability scales with the block length roughly as $2^{-\sqrt{n}}$, which gives a characterization of the \emph{error exponent} regime. Furthermore, the gap to capacity scales with the block length as $n^{-1/\mu}$, where the \emph{scaling exponent} $\mu$ depends on the transmission channel. Bounds on $\mu$ are provided first in \cite{HAU14}, then in \cite{GB14} and finally in \cite{MHU15unif-ieeeit}. In particular, in this last work, it is proved that $\mu \le 4.714$ for any BMS channel and that $\mu\le 3.639$ for the special case of the binary erasure channel (BEC), which approaches the value $3.627$ computed heuristically for the BEC. In \cite{MHU15unif-ieeeit}, it is also shown that the error probability scales with the Bhattacharyya parameter $Z(W)$ of the channel roughly as $Z(W)^{\sqrt{n}}$, which means that polar codes are \emph{not affected by error floors}.

%%%%%%%%%%%%%%%%%%%%%%%%%%%%
\subsection{Capacity-Achieving Rate-Compatible Polar Codes for Degraded Channels} \label{ratecompdegr}
%%%%%%%%%%%%%%%%%%%%%%%%%%%%

Let $W_1, W_2, \ldots, W_K$ be a family of BMS channels with respective capacities $I(W_1), I(W_2), \ldots, I(W_K)$ ordered by degradation, i.e., $W_1 \succ W_2 \succ \ldots \succ W_K$. Consider the transmission of $k$ bits of information via the rate-compatible family of polar codes $\{\Cc_{1}, \Cc_{2}, \ldots, \Cc_{K}\}$ with incremental block lengths $\{n_1, n_2, \ldots, n_K\}$ and rates $\{R_1, R_2, \ldots, R_K\}$ such that $R_{\ell} < I(W_{\ell})$, for $\ell \in [K]$. This means that each of the codewords of $\Cc_{\ell}$ can be decomposed into the codewords of $\ell$ polar codes with block lengths $\{n_1, n_2, \ldots, n_{\ell}\}$. In the following, we revise the scheme of \cite{HongHuiMaric2015} to construct capacity-achieving rate-compatible polar codes in this degraded scenario. This scheme is based on a chaining construction, in which some positions of a block are repeated in the following block.

\noindent {\bf Design of the Scheme.}
Let $X_{\ell}^{1:n_{\ell}}$ be a vector with $n_{\ell}$ i.i.d. uniformly random components and define $U_{\ell}^{1:n_{\ell}} = X_{\ell}^{1:n_{\ell}} G_{n_\ell}$. For the moment, we assume that $n_{\ell}$ is a power of $2$. When discussing the complexity of the code construction, we will mention how to tackle the case in which $n_{\ell}$ is not a power of $2$. As in \eqref{eq:lhsets}, consider the set 
\begin{equation} \label{eq:lhsetsrc}
\begin{split}
\Lc_{X_{\ell} \mid Y_{\ell, j}} &= \{i \in [n_{\ell}] \colon Z(U_{\ell}^{(i)} \mid  U_{\ell}^{1:i-1}, Y_{\ell, j}^{1:n_{\ell}}) \le \delta_{n_{\ell}}\},
\end{split}
\end{equation}
where $Y_{\ell, j}^{1:n_{\ell}}$ denotes the output of $W_j$ when $X_{\ell}^{1:n_{\ell}}$ is transmitted. Hence, for $i\in \Lc_{X_{\ell} \mid Y_{\ell, j}}$, the bit $U_{\ell}^{(i)}$ is approximately a deterministic function of the previous bits $U_{\ell}^{1:i-1}$ and the channel output $Y_{\ell, j}^{1:n_{\ell}}$. This means that $U_{\ell}^{1:n_{\ell}}$ can be decoded in a successive fashion, given channel output $Y_{\ell, j}^{1:n_{\ell}}$ and the values $\{U_{\ell}^{(i)}\}_{i\in \Lc_{X_{\ell} \mid Y_{\ell, j}}^{\rm c}}$, where $\Lc_{X_{\ell} \mid Y_{\ell, j}}^{\rm c}$ denotes the complement of $\Lc_{X_{\ell} \mid Y_{\ell, j}}$. Let $\Ac_j^{(\ell)}$ be a subset of $\Lc_{X_{\ell} \mid Y_{\ell, j}}$ such that
\begin{equation}\label{eq:sizeA}
|\Ac_j^{(\ell)}| = n_{\ell} R_j.
\end{equation}
Note this subset exists because $R_{\ell} < I(W_{\ell})$ and, as in \eqref{eq:cardHL},
\begin{equation} \label{eq:cardHLrc}
\begin{split}
\lim_{n_{\ell}\to \infty}\frac{1}{n_{\ell}} \, | \Lc_{X_{\ell} \mid Y_{\ell, j}}|  &= I(W_j).\\
\end{split}
\end{equation}
As the family of channels is ordered by degradation, by Lemma 4.7 of \cite{Kor09thesis}, we have that, for any $\ell, j, j'\in [K]$ with $j \le j'$,
\begin{equation}\label{eq:Lcdegraded}
\Lc_{X_{\ell} \mid Y_{\ell, j}} \supseteq \Lc_{X_{\ell} \mid Y_{\ell, j'}}.
\end{equation}
Hence, we can choose the sets $\Ac_j^{(\ell)}$ so that, for $j \le j'$, 
\begin{equation}\label{eq:Acdegraded}
\Ac_j^{(\ell)} \supseteq \Ac_{j'}^{(\ell)}.
\end{equation}

Consider the first transmission, in which the vector $X_{1}^{1:n_{1}}$ is sent over the channel. Note that $|\Ac_1^{(1)}| = n_1 R_1 = k$, because of \eqref{eq:conditionrate} and \eqref{eq:sizeA}. We put the $k$ information bits into $\Ac_1^{(1)}$ and we set to $0$ the remaining positions of $U_{1}^{1:n_{1}}$.

Consider the $\ell$-th transmission, for $\ell \in \{2, \ldots, K\}$, in which the vector $X_{\ell}^{1:n_{\ell}}$ is sent over the channel. Define
\begin{equation}
\Ic^{(\ell)} = \bigcup_{j=1}^{\ell-1} \Ac_{\ell-1}^{(j)} \setminus \Ac_{\ell}^{(j)},
\end{equation}
and note that 
\begin{equation}\label{eq:cardIcj}
\begin{split}
|\Ic^{(\ell)}| &= \sum_{j=1}^{\ell-1} \left|\Ac_{\ell-1}^{(j)} \setminus \Ac_{\ell}^{(j)}\right|= \sum_{j=1}^{\ell-1} \left|\Ac_{\ell-1}^{(j)} \right| -\left| \Ac_{\ell}^{(j)}\right|\\
& \stackrel{\mathclap{\mbox{\footnotesize(a)}}}{=} \sum_{j=1}^{\ell-1} n_{j}R_{\ell-1} - \sum_{j=1}^{\ell-1}n_{j}R_{\ell} \\
&= n_{\ell}R_\ell + \sum_{j=1}^{\ell-1} n_{j}R_{\ell-1} -\sum_{j=1}^{\ell}n_{j}R_{\ell} \\
& \stackrel{\mathclap{\mbox{\footnotesize(b)}}}{=} n_{\ell}R_\ell \stackrel{\mathclap{\mbox{\footnotesize(c)}}}{=}  |\Ac_\ell^{(\ell)}|,
\end{split}
\end{equation}
where (a) uses \eqref{eq:sizeA}, (b) uses \eqref{eq:conditionrate}, and (c) uses again \eqref{eq:sizeA}. Because of \eqref{eq:cardIcj}, we repeat in $\Ac_{\ell}^{(\ell)}$ the values of the positions in $\Ic^{(\ell)}$ and we set to $0$ the remaining positions of $U_{\ell}^{1:n_{\ell}}$. The situation is schematically represented in Figure \ref{fig:degraded}. 
\begin{figure}[t] 
\centering 
\includegraphics[width=1.1\columnwidth, height = .8\columnwidth]{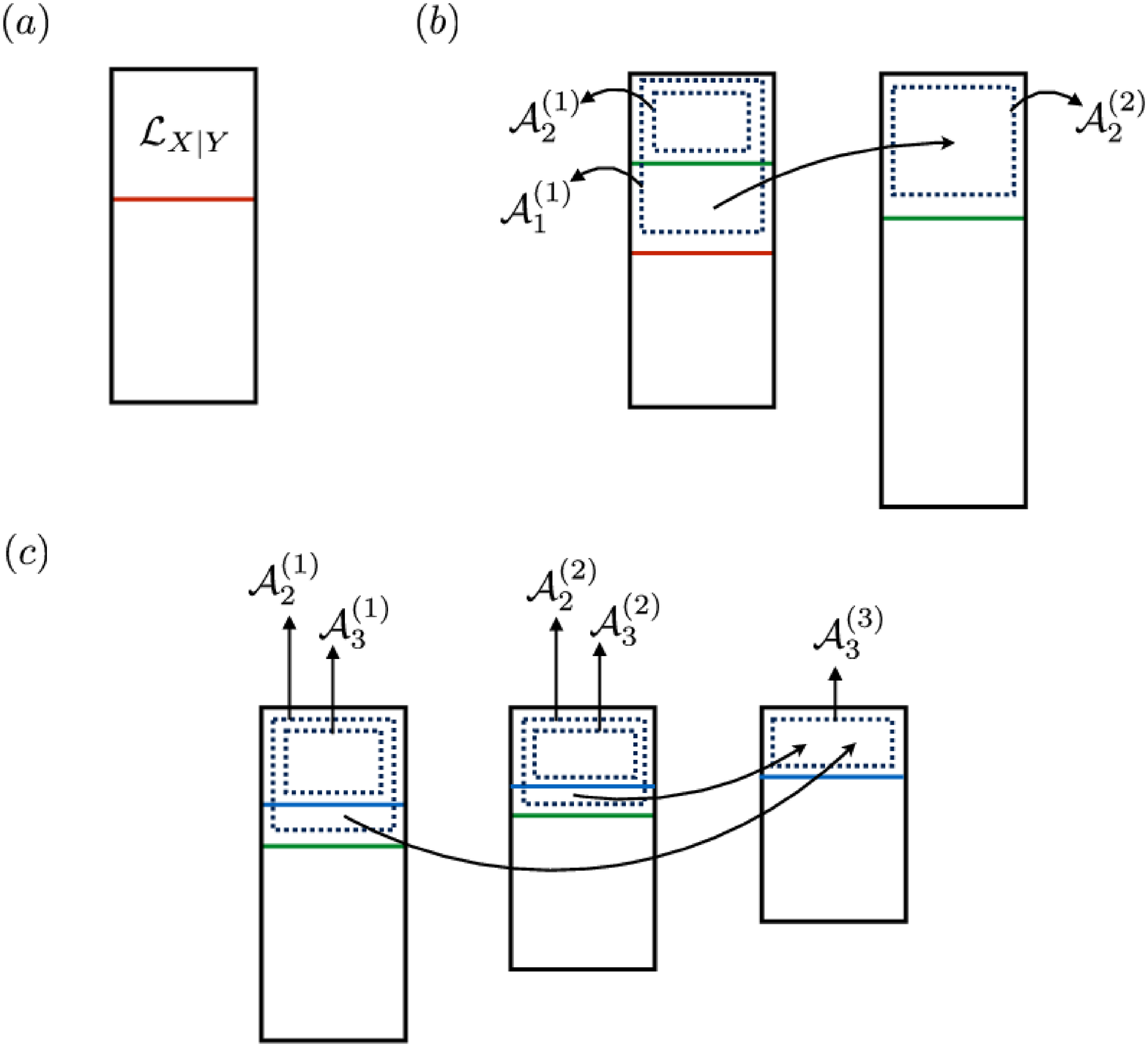}
\caption{(a) A simple graphical representation of the set $\mathcal{L}_{X \mid Y}$ defined in Section \ref{Polarcodes}.  The whole rectangle represents the index set $[n]$, and $\mathcal{L}_{X \mid Y}$ is the set of indices that are decodable in a successive fashion given the output. Note that $\mathcal{L}_{X \mid Y} $ is specified by the area above the red line. (b) Construction of rate-compatible polar codes for a family of degraded channels. Consider the $\ell$-th transmission for $\ell=2$ and note that the values of the indices in $A_{1}^{(1)} \setminus A_2^{(1)}$ are repeated in $\mathcal{A}_{2}^{(2)}$. In the smaller rectangle, the area above the red line represents $\mathcal{L}_{X_1 \mid Y_{1,1}}$ and the area above the green line represents $\mathcal{L}_{X_1 \mid Y_{1,2}}$. Also, in the larger rectangle the area above the green line represents $\mathcal{L}_{X_2 \mid Y_{2,2}}$.  (c) Consider the $\ell$-th transmission for $\ell=3$ and note that the values of the indices in $A_{2}^{(1)} \setminus A_3^{(1)}$ and $A_{2}^{(2)} \setminus A_3^{(2)}$ are repeated in $\mathcal{A}_{3}^{(3)}$.} 
\label{fig:degraded}
\end{figure} 

The decoding is performed ``backwards''. Assume that the transmission takes place over the channel $W_{\bar{\ell}}$ for some $\bar{\ell} \in [K]$. Let us first consider the case $\bar{\ell}=1$. The vector $X_{1}^{1:n_{1}}$ is transmitted over the channel and the receiver has access to $Y_{1, 1}^{1:n_{1}}$. Recall that $\Ac_1^{(1)}\subseteq \Lc_{X_{1} \mid Y_{1, 1}}$ and that the remaining positions of $U_{1}^{1:n_{1}}$ are frozen to $0$. Hence, the decoding of $U_{1}^{1:n_{1}}$ succeeds with high probability, and we reconstruct the $k$ information bits.

Consider now the case $\bar{\ell}>1$. The vectors $X_{1}^{1:n_{1}}, \ldots, X_{\bar{\ell}}^{1:n_{\bar{\ell}}}$ are transmitted over the channel and the receiver has access to $Y_{1, \bar{\ell}}^{1:n_{1}}, \ldots, Y_{\bar{\ell}, \bar{\ell}}^{1:n_{\bar{\ell}}}$. Recall that $\Ac_{\bar{\ell}}^{(\bar{\ell})}\subseteq \Lc_{X_{\bar{\ell}} \mid Y_{\bar{\ell}, \bar{\ell}}}$ and that the remaining positions of $U_{\bar{\ell}}^{1:n_{\bar{\ell}}}$ are frozen to $0$. Hence, the decoding of $U_{\bar{\ell}}^{1:n_{\bar{\ell}}}$ succeeds with high probability, and we reconstruct the values of the positions in $\Ic^{(\bar{\ell})}$. Recall that $\Ac_{\bar{\ell}}^{(\bar{\ell}-1)}\subseteq \Lc_{X_{\bar{\ell}-1} \mid Y_{\bar{\ell}-1, \bar{\ell}}}$, $\Ac_{\bar{\ell}-1}^{(\bar{\ell}-1)} \setminus \Ac_{\bar{\ell}}^{(\bar{\ell}-1)} \subseteq \Ic^{(\bar{\ell})}$, and that the remaining positions of $U_{\bar{\ell}-1}^{1:n_{\bar{\ell}-1}}$ are frozen to $0$. Hence, the decoding of $U_{\bar{\ell}-1}^{1:n_{\bar{\ell}-1}}$ succeeds with high probability, and we reconstruct the values of the positions in $\Ic^{(\bar{\ell}-1)}$. By repeating this procedure, we decode, in order, the vectors $U_{\bar{\ell}-2}^{1:n_{\bar{\ell}-2}}, \ldots, U_{1}^{1:n_{1}}$, and reconstruct the $k$ information bits.

The chaining construction described above is similar to the one developed for the wiretap channel in \cite{SaV13} and for the broadcast channel in \cite{MHSU14broad-ieeeit}. The analogy consists in the fact that the transmission is divided into several blocks and these blocks are chained together by repeating a suitable set of indices. The difference is that the current scheme suffers no rate loss, hence we do not require that the number of blocks goes large in order to achieve capacity. In the setting of this paper, the number of blocks depends on the channel over which the transmission takes place and  is upper bounded by $K$.

\noindent {\bf Complexity and Performance Analysis.} Consider the $\ell$-th transmission, for $\ell\in [K]$. Let us first discuss the \emph{code construction}. The problem consists in finding the sets $\{\Ac^{(\ell)}_{j}\}_{j\ge \ell}$. If the block length $n_{\ell}$ is a power of 2, nothing changes with respect to the standard transmission over a BMS channel and, for any $j\ge \ell$, the set $\Ac^{(\ell)}_{j}$ can be found in $\Theta(n_{\ell})$ by using the techniques in \cite{TV13con, RHTT}. Consequently, the complexity of the code construction is linear in the block length. However, given the rates $\{R_1, R_2, \ldots, R_K\}$, the block lengths need to satisfy \eqref{eq:conditionrate}. For this reason, $n_\ell$ might not be a power of 2. This issue is solved by Theorem 1 of \cite{HongHuiMaric2015}, where it is proved that, for any puncturing fraction and for any BMS channel, there exists a sequence of capacity-achieving punctured polar codes. The proof uses a random puncturing argument and models the punctured positions as outputs of the BEC. More concretely, let $m_\ell$ be the smallest power of 2 larger than $n_\ell$. The idea is to build a polar code of block length $m_\ell$ and puncture $m_\ell - n_\ell$ positions. The puncturing pattern is chosen uniformly at random, and we exploit the knowledge of the locations of the punctured bits in order to construct the polar code of block length $m_{\ell}$. In particular, the punctured bits are modelled as the output of a BEC with erasure probability $1$ and the remaining positions as the output of the channel $W_\ell$. Again, the code can be constructed with complexity linear in $m_{\ell}$, hence linear in $n_{\ell}$~\cite{TV13con, RHTT}. Note that we can repeat this construction for several puncturing patterns and pick the one that yields the best error performance (i.e., that minimizes the sum of the smallest $n_{\ell}R_{\ell}$ Bhattacharyya parameters of the synthetic channels).

The \emph{encoding and decoding complexity} remains of $\Theta(n_\ell \log_2 n_\ell)$, as in the standard polar  construction. Similarly, the \emph{error performance} is also the same as that of standard polar codes, which is discussed at the end of Section \ref{Polarcodes}.

%%%%%%%%%%%%%%%%%%%%%%%%%%%%
\section{Main Result} \label{sec:mainres}
%%%%%%%%%%%%%%%%%%%%%%%%%%%%

Let us state the main result of this paper. 

\begin{MyTheorem}\label{th:ratecompgen}
For any family of BMS channels $W_1, W_2, \ldots, W_K$ with capacities $I(W_1)>I(W_2)>\ldots>I(W_K)$, there exists a sequence of rate-compatible polar codes that is capacity-achieving.
\end{MyTheorem}

\noindent {\bf Design of the Scheme.} If the family of channels $W_1, W_2, \ldots, W_K$ is ordered by degradation, we follow the scheme outlined in Section \ref{ratecompdegr}. Assume now that the degradation condition does not hold. Hence, condition \eqref{eq:Lcdegraded} might not hold, which also implies that  \eqref{eq:Acdegraded} might not hold. By analyzing carefully the argument of Section \ref{ratecompdegr}, we notice that, in order for the scheme to work, we need that \eqref{eq:Acdegraded} holds for any $\ell \in [K]$, for $j\ge \ell$ and for $j'\ge j$. Consequently, our strategy is to construct polar codes such that \eqref{eq:Acdegraded} holds for the desired values of $\ell$, $j$, and $j'$.  

Let us first consider the case $K=2$. The condition to be fulfilled is $\Ac_1^{(1)} \supseteq \Ac_{2}^{(1)}$. Define $\Lc_{X_{1} \mid Y_{1, 1}}$ and $\Lc_{X_{1} \mid Y_{1, 2}}$ as in \eqref{eq:lhsetsrc}. If $\Lc_{X_{1} \mid Y_{1, 1}} \supseteq \Lc_{X_{1} \mid Y_{1, 2}}$, then such a condition can be satisfied automatically and no further action is needed. However, in general, we have that $\Lc_{X_{1} \mid Y_{1, 1}} \not \supseteq \Lc_{X_{1} \mid Y_{1, 2}}$. To handle this case, we use the ideas developed in Section V of \cite{HRunipol}. Note that in \cite{HRunipol} a family of channels with the \emph{same} capacity is considered and the aim is to construct universal polar codes, i.e., codes that are reliable for every channel in the family. Here, the setting is different, as $I(W_1) > I(W_2)$. However, a similar strategy will prove successful. 

Define $\Dc = \Lc_{X_{1} \mid Y_{1, 2}} \setminus \Lc_{X_{1} \mid Y_{1, 1}}$ and let $\Dc'$ be a subset of $\Lc_{X_{1} \mid Y_{1, 1}} \setminus \Lc_{X_{1} \mid Y_{1, 2}}$ such that $|\Dc'| = |\Dc|$. By using \eqref{eq:cardHLrc} and that $I(W_1) > I(W_2)$, we have that, for $n_1$ sufficiently large, $|\Lc_{X_{1} \mid Y_{1, 1}}| \ge |\Lc_{X_{1} \mid Y_{1, 2}}|$. Hence, such a subset $\Dc'$ exists. Let us order the indices in $\Dc$ and $\Dc'$ so that $\Dc = \{d_1, \ldots, d_{|\Dc|}\}$ with $d_1 < \ldots < d_{|\Dc|}$ and $\Dc' = \{d_1', \ldots, d_{|\Dc|}'\}$ with $d_1' < \ldots < d_{|\Dc|}'$. Set $d_0 = d_0' = 0$.

The idea is to perform one further step of polarization by combining two i.i.d. copies of $U_{1}^{1:n_{1}}$ into a vector of length $2n_{1}$. In a nutshell, we polarize the positions indexed by $\Dc$ of the first copy with the positions indexed by $\Dc'$ of the second copy and we leave the remaining positions unchanged. More specifically, we proceed as follows first for $i=1$, then $i=2$, and so on until $i= |\Dc|$: we repeat the values of the first copy from index $d_{i-1}+1$ until index $d_i-1$; we repeat the values of the second copy from index $d_{i-1}'+1$ until index $d_i'-1$; we put the XOR of the bit in position $d_i$ of the first copy with the bit in position $d_i'$ of the second copy; and we repeat the bit in position $d_i'$ of the second copy. Finally, we repeat the remaining positions (if any) of the first copy, and then the remaining positions (if any) of the second copy. The situation is schematically represented in Figure \ref{fig:nondegraded}.  Note that the set $\Dc$ contains the positions that are \emph{good} (i.e., they can be decoded in a successive fashion) for $W_2$, but not for $W_1$. Similarly, the set $\Dc'$ contains the positions that are \emph{good} for $W_1$, but not for $W_2$. By combining the positions of $\Dc$ with the positions of $\Dc'$ we obtain $|2\Dc|$ positions such that one half of them is good for both $W_1$ and $W_2$ and the other half is good for neither $W_1$ nor $W_2$.

\begin{figure}[t] 
\centering 
\includegraphics[width=1.1\columnwidth, height = .8\columnwidth]{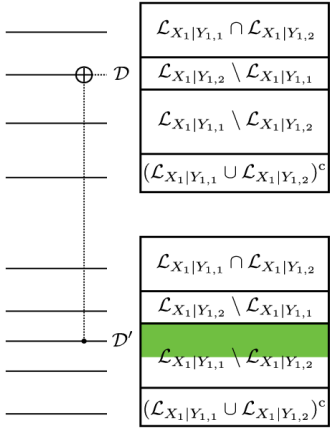}
\vspace{-1.1cm}

\caption{A graphical representation of one further step of polarization required when the family of channels is not ordered by degradation. By performing sufficiently many such polarization steps, we construct capacity-achieving rate-compatible polar codes for any family of BMS channels. } 
\label{fig:nondegraded}
\vspace{-.5cm}
\end{figure} 

Consider the vector of length $2n_1$ obtained with the procedure above. Clearly, the fraction of positions that are good for $W_1$ is the same and it tends to $I(W_1)$. Similarly, the fraction of positions that are good for $W_2$ is the same and it tends to $I(W_2)$. Furthermore, the fraction of positions that are good for $W_2$ but not for $W_1$ is halved. By repeating this procedure for $t$ times, we obtain a code of block length $n_1\cdot 2^t$ such that the fraction of positions that are good for $W_2$ but not for $W_1$ is at most $2^{-t} |\Dc|/n_1< 2^{-t}$. Consequently, the fraction of positions that are good both for $W_1$ and $W_2$ is at least $I(W_2)-2^{-t}$. By choosing $t$ sufficiently large, we can find a subset $\Ac_{2}^{(1)}$ of those positions such that $|\Ac_{2}^{(1)}| = 2^t n_1 R_2$. Hence, there exists a subset $\Ac_{1}^{(1)}$ of positions that are good for $W_1$ such that $|\Ac_{1}^{(1)}| = 2^t n_1 R_1$ and $\Ac_1^{(1)} \supseteq \Ac_{2}^{(1)}$, as required.

Let us now go back to our original problem of satisfying \eqref{eq:Acdegraded} for any $\ell \in [K]$, for $j\ge \ell$ and for $j'\ge j$. Consider the $\ell$-th transmission for $\ell\in [K]$. Note that, in order to satisfy \eqref{eq:Acdegraded} for $j\ge \ell$ and $j'\ge j$, it suffices to have that
\begin{equation}\label{eq:newcond}
\Ac^{(\ell)}_{\ell}\supseteq \Ac^{(\ell)}_{\ell+1}\supseteq \ldots \supseteq \Ac^{(\ell)}_{K}.
\end{equation}
To fulfill this last requirement, we proceed as follows first for $j=\ell$, then $j=\ell+1$, and so on until $j=K-1$. Define $\Lc_{X_{\ell} \mid Y_{\ell, j}}$ and $\Lc_{X_{\ell} \mid Y_{\ell, j+1}}$ as in \eqref{eq:lhsetsrc}. If $\Lc_{X_{\ell} \mid Y_{\ell, j}} \supseteq \Lc_{X_{\ell} \mid Y_{\ell, j+1}}$, then $\Ac^{(\ell)}_{j}\supseteq \Ac^{(\ell)}_{j+1}$ is satisfied automatically and no further action is needed. Otherwise, we perform $t$ further steps of polarization as described in the previous paragraph. In this way, we increase the block length of the code by a factor of $2^t$, which means that the cardinalities of both $\Ac^{(\ell)}_{j}$ and $\Ac^{(\ell)}_{j+1}$ increase by a factor $2^t$. In addition, the fraction of positions that are good both for $W_{j}$ and $W_{j+1}$ is at least $I(W_{j+1})-2^{-t}$. Therefore, by choosing $t$ sufficiently large, we can ensure that $\Ac^{(\ell)}_{j} \supseteq \Ac^{(\ell)}_{j+1}$, as requested.

By adding these additional polarization steps, condition \eqref{eq:newcond} is satisfied. Consequently, by repeating the same argument of \cite{HongHuiMaric2015}, we obtain capacity-achieving rate-compatible polar codes for any family of channels and Theorem \ref{th:ratecompgen} is  proved.

\noindent {\bf Complexity and Performance Analysis.} Consider the $\ell$-th transmission, for $\ell\in [K]$. The \emph{code construction} is performed in the same way as in the degraded case analyzed in Section \ref{ratecompdegr}: on the one hand, if $n_{\ell}$ is a power of 2, nothing changes with respect to the standard transmission over a BMS channel; on the other hand, if $n_{\ell}$ is not a power of 2, we construct punctured polar codes. As a result, the complexity is 
$\Theta(n_{\ell})$.

Let us now consider the \emph{encoding and decoding complexity}. Recall that we require $t$ further polarization steps to ensure that $\Ac^{(\ell)}_{j} \supseteq \Ac^{(\ell)}_{j+1}$, for a fixed $j\in \{\ell, \ell+1, \ldots, K-1\}$. Hence, the total number of further polarization steps is at most $(K-1)t$, which yields an additional encoding and decoding complexity of $\Theta((K-1)t \cdot n_\ell)$. Note that $t$ and $K$ are fixed, while $n_{\ell}$ grows large. As a result, the overall encoding and decoding complexity remains of $\Theta(n_\ell \log_2 n_\ell)$, as in the standard polar coding construction.

Finally, let us discuss the \emph{error performance} for the $\ell$-th transmission, $\ell\in [K]$. The removal of the degradation assumption comes at a cost: we need to increase the block length of the code and we suffer a rate loss. The block length $n_{\ell}$ increases by a factor of at most $2^{(K-1)t}$, as the total number of further polarization steps is at most $(K-1)t$. The rate loss represents the gap between $R_{\ell}$ and $I(W_{\ell})$. For every $j$ such that we need to ensure $\Ac^{(j)}_{\ell-1}\supseteq \Ac^{(j)}_{\ell}$, we suffer a rate loss of at most $2^{-t}$. Hence, the total rate loss is at most $(K-1)2^{-t}$. This upper bound does not depend on the family of channels $W_1, W_2, \ldots, W_K$. Hence, given a requirement on the gap to capacity, the value of $t$ can be chosen prior to code design. Furthermore, let us point out that there is a trade-off between the rate loss and the increase in the block length: larger values of $t$ make the rate loss smaller and the block length bigger. 

The increase in the block length and the rate loss does not depend on the block length $n_{\ell}$, but only on $K$ and $t$. Furthermore, as $n_{\ell}$ grows large, the parameters $K$ and $t$ are constants, which implies that the scaling of $2^{-\sqrt{n_{\ell}}}$ and of $2^{-\sqrt{n_{\ell}\cdot 2^{-(K-1)t}}}$ is roughly the same. Consequently, the error probability scales with the block length roughly as $2^{-\sqrt{n_{\ell}}}$, which gives the same characterization as standard polar codes in terms of the \emph{error exponent} regime. For the same reason, the error probability scales with the Bhattacharyya parameter $Z(W_\ell)$ of the channel as $Z(W_\ell)^{\sqrt{n_{\ell}}}$, which means that the proposed scheme is \emph{not affected by error floors}. However, the \emph{scaling exponent is affected} by the removal of the degradation assumption. This fact is clarified by the next theorem whose proof immediately follows.

\begin{MyTheorem}\label{th:scalexp}
Consider the transmission of a sequence of rate-compatible polar codes over the family of channels $W_1, W_2, \ldots, W_K$. Let $\mu$ be the scaling exponent for the transmission of a polar code over $W_{\ell}$, for $\ell\in [K]$. Then, the scaling exponent associated to the $\ell$-th transmission of the rate-compatible family of polar codes is upper bounded by $\mu +K-1$. 
\end{MyTheorem}

\begin{proof}
By definition of scaling exponent, the block length of a polar code with gap to capacity $\delta$ is upper bounded by $c/\delta^\mu$, for some constant $c$. Recall that the block length associated to the $\ell$-th transmission of the rate-compatible family of polar codes increases by a multiplicative factor of at most $2^{(K-1)t}$. Furthermore, the gap to capacity increases by an additive factor of at most $(K-1)2^{-t}$. Set $(K-1)2^{-t} = \delta$. Then, the total gap to capacity is upper bounded by $2\delta$ and the total block length is upper bounded by 
\begin{equation*}
\frac{c}{\delta^\mu} \cdot 2^{(K-1)t} = \frac{c}{\delta^\mu} \cdot \frac{(K-1)^{K-1}}{\delta^{K-1}} = \frac{2^{K+\mu-1}c(K-1)^{K-1}}{(2\delta)^{K+\mu-1}}, 
\end{equation*}
which proves the claim, as $2^{K+\mu-1}c(K-1)^{K-1}$ is a constant. 
\end{proof}
\vspace{-.3cm}

%%%%%%%%%%%%%%%%%%%%%%%%%%%%
\section{Conclusions} \label{sec:concl}
%%%%%%%%%%%%%%%%%%%%%%%%%%%%

We construct rate-compatible polar codes that achieve the capacity of \emph{any family of BMS channels}. To do so, we follow the lead of \cite{HongHuiMaric2015, LiTse2015}, where capacity-achieving rate-compatible polar codes are proposed for a \emph{degraded family of channels}. In order to remove the degradation condition, we take inspiration from the techniques to devise universal polar codes \cite{HRunipol, SaWa16}. The idea is to perform further polarization steps so that the sets of polar indices can be aligned as in the case of a degraded family of channels.

The proposed scheme requires that the transmitter knows over which family of channels the transmission is taking place. Our solution has the usual attractive features of polar codes as for the complexity and the error performance. A disadvantage is that the additional polarization steps yield a larger scaling exponent. Another drawback, from the implementation perspective, is that the backward decoding requires all the polar blocks to be stored until the last retransmission is received, thereby increasing the amount of buffering that is needed. The improvement of the scaling exponent and the reduction of the buffering size are left as open problems. 

We limit our study to channels that have a binary input alphabet and are symmetric. However, there seems to be no essential difficulty in extending our results to channels that have an arbitrary input alphabet and are asymmetric, by exploiting the ideas in \cite{STA09, MHU14asymm-ieeeit}.

\section*{Acknowledgment}
The work of M. Mondelli is supported by grant No. 200021\_166106 of the Swiss National Science Foundation and by the Dan David Foundation.

%%%%%%%%%%%%%%%%%%%%%%%%%%%%%%%%%%%%%%%%%%%
%%%%%%%%%%%%%%%%%%%%%%%%%%%%%%%%%%%%%%%%%%%
\bibliographystyle{IEEEtran}
\bibliography{lth,lthpub}
%%%%%%%%%%%%%%%%%%%%%%%

%%%%%%%%%%%%%%%%%%%%%%%%%%%%%%%%%%%%%%%%%%%%
\end{document}